\documentclass{llncs}

\usepackage{amssymb}
\usepackage{graphicx}

\usepackage[latexsym]{}
\usepackage{amssymb}
\usepackage{amsmath}
\usepackage{graphicx}
\usepackage{epsfig}
\usepackage{url}
\usepackage{algorithm}
\usepackage{algorithmic}
\usepackage{stmaryrd}
\usepackage{qtree}

\usepackage{multicol}

\newcommand{\spqlex}{\text{SPARQL}^{\mbox{\tiny{NEX}}}} 
\newcommand{\spqlexs}{\text{SPARQL}^{\mbox{\tiny{NEX}}}_{\mbox{\tiny{safe}}}}  
\newcommand{\spqldiff}{\text{SPARQL}^{\mbox{\tiny{DIFF}}}}

\newcommand{\blank}{\operatorname{blank}}

\newcommand{\dg}{\operatorname{dg}}
\newcommand{\name}{\operatorname{name}}
\newcommand{\names}{\operatorname{names}}
\newcommand{\graph}{\operatorname{gr}}
\newcommand{\dom}{\operatorname{dom}}

\newcommand{\false}{\operatorname{\emph{false}}}
\newcommand{\true}{\operatorname{\emph{true}}}
\newcommand{\error}{\operatorname{\emph{error}}}

\newcommand{\AAND}{\operatorname{AND}}
\newcommand{\UNION}{\operatorname{UNION}}
\newcommand{\OPT}{\operatorname{OPT}}
\newcommand{\OPTIONAL}{\operatorname{OPT}}
\newcommand{\FILTER}{\operatorname{FILTER}}
\newcommand{\MINUS}{\operatorname{MINUS}}
\newcommand{\DIFF}{\operatorname{DIFF}}
\newcommand{\GRAPH}{\operatorname{GRAPH}}

\newcommand{\bound}{\operatorname{bound}}
\newcommand{\EXCEPT}{\operatorname{EXCEPT}}

\newcommand{\NEX}{\operatorname{NOT-EXISTS}}

\newcommand{\NEG}{\operatorname{NEG}}

\def\lojoin{\hbox{\raise -.2em\hbox to-.32em{$\urcorner$} \hbox to-.08em{$\lrcorner$} $\Join$}\,}
\newcommand{\ev}[2]{\llbracket #1 \rrbracket_{#2}} 
\newcommand{\eval}[3]{\llbracket #1 \rrbracket^{#2}_{#3}}

\newcommand{\var}{\operatorname{var}}
\newcommand{\svar}{\operatorname{svar}}

\newcommand{\cardi}[2]{\operatorname{card}(#1,#2)} 

\usepackage{url}

\begin{document}
\pagestyle{empty}

\title{Negation in SPARQL}
\author{Renzo Angles\inst{1,3} \and Claudio Gutierrez\inst{2,3}}
\institute{
Dept. of Computer Science, Universidad de Talca, Chile\\
\and 
Dept. of Computer Science, Universidad de Chile, Chile
\and
Center for Semantic Web Research
\\
}

\maketitle

\begin{abstract}
This paper presents a thorough study of negation in SPARQL.
The types of negation supported in SPARQL are identified and their main features discussed.
Then, we study the expressive power of the corresponding negation operators.
At this point, we identify a core SPARQL algebra which could be used instead of the W3C SPARQL algebra.   
Finally, we analyze the negation operators in terms of their compliance with elementary axioms of set theory.
\end{abstract}


\section{Introduction}

The notion of negation has been largely studied in database query
languages,  mainly due to their implications in aspects of expressive
power and computational complexity \cite{91282,91284,50660,90947,91285}. 
There have been proposed several and distinct types of negation,  and
it seems difficult to get agreement about a standard one. 
Such heterogeneity comes from intrinsic properties and semantics of 
each language, features that determine its ability to support specific
type(s) of negation(s). For instance, rule-based query languages
(e.g. Datalog) usually  implement ``negation by failure'', 
algebraic query languages (e.g. the relational algebra) usually
include a Boolean-like difference operator, and SQL-like query
languages support several versions of negation, and operators 
that combine negation with subqueries (e.g. NOT EXISTS).      



 The case of SPARQL is no exception. SPARQL 1.0 did not have any
explicit form of negation in its syntax, but could simulate negation
by failure. For the next version, SPARQL 1.1, the incorporation of
negation  generated a lot of debate.
As result, SPARQL 1.1 provides four types of negation: 
negation of filter constraints, by using the Boolean NOT operator; 
negation as failure, implemented as the combination of an optional graph pattern and the bound operator; 
difference of graph patterns, expressed by the $\MINUS$ operator;
and
existential negation of graph patterns, expressed by the $\NEX$ operator.  

The main positive and negative aspects of these types of negation are
remarked in the SPARQL specifications, and more detailed discussions
are available in the records of the SPARQL working
group\footnote{\url{https://www.w3.org/2009/sparql/wiki/Design:Negation}}. However,
to the best of our knowledge, there exists no formal study about 
negation in SPARQL. This is the goal we tackle throughout this document.

 In this paper we present a thorough study of negation in SPARQL.
First, we formalize the syntax and semantics of the different types of negation supported in SPARQL, and discuss their main features.
 Then, we study the relationships (in terms of expressive power) among the
negation operators, first at the level of the SPARQL algebra, and
subsequently at the level of SPARQL graph patterns.
 Finally, we present a case-by-case analysis of the negation operations with respect to their compliance with elementary axioms of set theory.  

It would be good to have a simple version of negation operators that conform to a known intuition. This is the other main contribution of this paper. 
  First we introduce the DIFF operator as another way of expressing the negation of graph patterns (it is the SPARQL version of the EXCEPT operator of SQL). 
  The semantics of DIFF is based on a simple difference operator introduced at the level of the SPARQL algebra.
 With this new difference operator, we show that one can define a core algebra (i.e. projection, selection, join, union and simple difference) which is able to express the W3C SPARQL algebra. 
 We show that this algebra is also able to define the SPARQL operators, thus it can be considered a sort of ``core'' SPARQL algebra. 
 Additionally, we show that the DIFF operator behaves well regarding its compliance with well known axioms of set theory.  

\paragraph{Organization of the paper.}
The syntax and semantics of SPARQL graph patterns are presented in Section 2.
In Section 3, we discuss the main characteristics of the types of negation supported by SPARQL.
The expressive power of the negation operators is studied in Section 4.
In Section 5, we analyze set-theoretic properties of two negation operators.
Finally, some conclusions are presented in Section 6.


\section{SPARQL graph patterns}
\label{sec:patterns}
The following definition of SPARQL graph patterns is based on the formalism used in \cite{10160}, but in agreement with the W3C SPARQL specifications \cite{10155,90699}.
Particularly, graph patterns will be studied assuming bag semantics (i.e. allowing duplicates in solutions).  

\paragraph{RDF graphs.}
Assume two disjoint infinite sets $I$ and $L$, called IRIs and literals respectively. 
An \emph{RDF term} is an element in the set $T = I \cup L$\footnote{In addition to $I$ and $L$, RDF and SPARQL consider a domain of anonymous resources called blank nodes.  The occurrence of blank nodes introduces several issues that are not discussed in this paper. Based on the results presented in \cite{91033}, we avoid the use of blank nodes assuming that their absence does not largely affect the results presented in this paper.}.
 An \emph{RDF triple} is a tuple $(v_1,v_2,v_3) \in I \times I \times T$ where $v_1$ is the \emph{subject}, $v_2$ the
\emph{predicate} and $v_3$ the \emph{object}.  
An \emph{RDF Graph} (just graph from now on) is a set of RDF triples.  
The \emph{union} of graphs, $G_1 \cup G_2$, is the set theoretical
union of their sets of triples.  
 Additionally, assume the existence of an infinite set $V$ of variables disjoint from $T$.
We will use $\var(\alpha)$ to denote the set of variables occurring in the structure $\alpha$. 

A \emph{solution mapping} (or just \emph{mapping} from now on) is a partial function $\mu : V \to T$ where the domain of $\mu$, $\dom(\mu)$, is the subset of $V$ where $\mu$ is defined. 
The \emph{empty mapping}, denoted $\mu_0$, is the mapping satisfying that 
$\dom(\mu_0)= \emptyset$. 
Given $?X \in V$ and $c \in T$, we use $\mu(?X) = c$ to denote the solution mapping variable $?X$ to term $c$. 
 Similarly, $\mu_{?X \to c}$ denotes a mapping $\mu$ satisfying that $\dom(\mu)=\{?X\}$ and $\mu(?X) = c$.
 Given a finite set of variables $W \subset V$, the restriction of a mapping $\mu$ to $W$, denoted $\mu_{|W}$, is a mapping $\mu'$ satisfying that  $\dom(\mu') = \dom(\mu) \cap W$ and $\mu'(?X) = \mu(?X)$ for every $?X \in \dom(\mu) \cap W$.
 Two mappings $\mu_1, \mu_2$ are \emph{compatible}, denoted $\mu_1 \sim \mu_2$, when for all $?X \in \dom(\mu_1) \cap \dom(\mu_2)$ it satisfies that $\mu_1(?X)=\mu_2(?X)$, i.e., when $\mu_1 \cup \mu_2$ is also a mapping.
Note that two mappings with disjoint domains are always compatible, and that the empty mapping $\mu_0$ is compatible with any other mapping.

A \emph{selection formula} is defined recursively as follows: 
(i) If $?X,?Y \in V$ and $c \in I \cup L$ then $(?X = c)$, $(?X = ?Y)$ and $\bound(?X)$ are atomic selection formulas;
(ii) If $F$ and $F'$ are selection formulas then $(F \land F')$, $(F \lor F')$ and $\neg (F)$ are boolean selection formulas.
 The evaluation of a selection formula $F$ under a mapping $\mu$, denoted $\mu(F)$, is defined in a three-valued logic (i.e. with values $\true$, $\false$, and $\error$) as follows:
\begin{itemize}
\item If $F$ is $?X = c$ and $?X \in \dom(\mu)$, then $\mu(F) = \true$ when $\mu(?X) = c$ and $\mu(F) = \false$ otherwise. If $?X \notin \dom(\mu)$ then $\mu(F) = \error$.
\item If $F$ is $?X = ?Y$ and $?X,?Y \in \dom(\mu)$, then $\mu(F) = \true$ when $\mu(?X) = \mu(?Y)$ and $\mu(F) = \false$ otherwise. If either $?X \notin \dom(\mu)$ or $?Y \notin \dom(\mu)$ then $\mu(F) = \error$.
\item If $F$ is $\bound(?X)$ and $?X \in \dom(\mu)$ then $\mu(F) = \true$ else $\mu(F) = \false$.
\item If $F$ is a complex selection formula then it is evaluated following the three-valued logic presented in Table \ref{table:tvl}.
\end{itemize}

\begin{table}[t]
\centering
\begin{tabular}{ccccccc}
\hline
$p$ & & $q$ & & $p \land q$ & & $p \lor q$ \\ 
\hline
$\true$ && $\true$ && $\true$ && $\true$ \\
$\true$ && $\false$ && $\false$ && $\true$ \\
$\true$ && $\error$ && $\error$ && $\true$ \\
$\false$ && $\true$ && $\false$ && $\true$ \\
$\false$ && $\false$ && $\false$ && $\false$ \\
$\false$ && $\error$ && $\false$ && $\error$ \\
$\error$ && $\true$ && $\error$ && $\true$ \\
$\error$ && $\false$ && $\false$ && $\error$ \\
$\error$ && $\error$ && $\error$ && $\error$ \\
\hline 
\\
\end{tabular}
\begin{tabular}{ccc}
\hline
$p$ & & $\neg p$ \\ 
\hline
$\true$ && $\false$ \\
$\false$ && $\true$ \\
$\error$ && $\error$ \\
\hline
\end{tabular}
\caption{Three-valued logic for evaluating selection formulas.} 
\label{table:tvl}
\end{table}

A \emph{multiset} (or \emph{bag}) of solution mappings is an unordered collection in which each solution mapping may appear more than once.
A multiset will be represented as a set of solution mappings, each one annotated with a positive integer which defines its acity (i.e. its cardinality).
 We use the symbol $\Omega$ to denote a multiset and $\cardi{\mu}{\Omega}$ to denote the cardinality of the mapping $\mu$ in the multiset $\Omega$. 
In this sense, it applies that $\cardi{\mu}{\Omega} = 0$ when $\mu \notin \Omega$.
We use $\Omega_0$ to denote the multiset $\{ \mu_0 \}$ such that $\cardi{\mu_0}{\Omega_0} > 0$ ($\Omega_0$ is called the join identity).
The domain of a solution mapping $\Omega$ is defined as $\dom(\Omega) = \bigcup_{\mu \in \Omega} \dom(\mu)$.  

\paragraph{W3C SPARQL algebra.}
Let $\Omega_1,\Omega_2$ be multisets of mappings, $W$ be a set of variables and $F$ be a selection formula.
The \emph{SPARQL algebra for multisets of mappings} is composed of the operations of projection, selection, join, difference, left-join, union and minus, defined respectively as follows:
\begin{itemize}

\item  
$\pi_W(\Omega_1) = \{ \mu' \mid \mu \in \Omega_1, \mu' = \mu_{|W} \}$\\ 
where 
$\cardi{\mu'}{\pi_W(\Omega_1)} = \sum_{\mu' = \mu_{|W}} \cardi{\mu}{\Omega_1}$ 

\item   
$\sigma_F(\Omega_1) = \{ \mu \in \Omega_1 \mid \mu(F) = \true \}$ \\ 
where $\cardi{\mu}{\sigma_F(\Omega_1)} = \cardi{\mu}{\Omega_1}$ 

\item 
$\Omega_1 \Join \Omega_2 = \{ \mu = (\mu_1 \cup \mu_2) \mid \mu_1 \in \Omega_1, \mu_2 \in \Omega_2, \mu_1 \sim \mu_2 \}$ \\
where
$\cardi{\mu}{\Omega_1 \Join \Omega_2} = \sum_{\mu = (\mu_1 \cup \mu_2)} \cardi{\mu_1}{\Omega_1}~\times~\cardi{\mu_2}{\Omega_2}$

\item 
$\Omega_1 \setminus_F \Omega_2 = \{ \mu_1 \in \Omega_1 \mid \forall \mu_2 \in \Omega_2, (\mu_1 \nsim \mu_2) \lor (\mu_1 \sim \mu_2 \land (\mu_1 \cup \mu_2)(F) = \false ) \}$ \\
where
$\cardi{\mu_1}{\Omega_1 \setminus_F \Omega_2} = \cardi{\mu_1}{\Omega_1}$

\item 
$\Omega_1 \cup \Omega_2 = \{ \mu \mid \mu \in \Omega_1 \lor \mu \in \Omega_2 \}$ \\
where
$\cardi{\mu}{\Omega_1 \cup \Omega_2} = \cardi{\mu}{\Omega_1} + \cardi{\mu}{\Omega_2}$ 

\item 
$\Omega_1 - \Omega_2 = \{ \mu_1 \in \Omega_1 \mid \forall \mu_2 \in \Omega_2, \mu_1 \nsim \mu_2 \lor \dom(\mu_1) \cap \dom(\mu_2) = \emptyset \}$ \\
where
$\cardi{\mu_1}{\Omega_1 - \Omega_2} = \cardi{\mu_1}{\Omega_1}$ 

\item
$\Omega_1 \lojoin_F \Omega_2 = \sigma_F(\Omega_1 \Join \Omega_2) \cup (\Omega_1 \setminus_F \Omega_2)$ \\
where
$\cardi{\mu}{\Omega_1 \lojoin_F \Omega_2} = \cardi{\mu}{\sigma_F(\Omega_1 \Join \Omega_2)} + \cardi{\mu}{\Omega_1 \setminus_F \Omega_2}$ 

\end{itemize}

\paragraph{Syntax of SPARQL graph patterns.}
A SPARQL \emph{graph pattern} is defined recursively as follows:
\begin{itemize}

\item 
A tuple from $(I \cup L \cup V) \times (I \cup V) \times (I \cup L \cup V)$ is a graph pattern called a \emph{triple pattern}.
\footnote{We assume that any triple pattern contains at least one variable.} 

\item If $P_1$ and $P_2$ are graph patterns then 
$( P_1 \AAND P_2 )$, 
$( P_1 \UNION P_2 )$,\\ 
$( P_1 \OPTIONAL P_2 )$, 
$( P_1 \MINUS P_2 )$ and 
$( P_1 \NEX P_2 )$ 
are graph patterns.

\item If $P_1$ is a graph pattern and $C$ is a filter constraint (as defined below) then $(P_1 \FILTER C)$ is a graph pattern.
\end{itemize}

A \emph{filter constraint} is defined recursively as follows:
(i) If $?X,?Y \in V$ and $c \in I \cup L$ then $(?X = c)$, $(?X = ?Y)$ and $\bound(?X)$ are \emph{atomic filter constraints};
(ii) If $C_1$ and $C_2$ are filter constraints then 
$(!C_1)$, $(C_1~||~C_2)$ and $(C_1~\&\&~C_2)$ 
are \emph{complex filter constraints}.
Given a filter constraint $C$, we denote by $f(C)$ the selection formula obtained from $C$. Note that there exists a simple and direct translation from filter constraints to selection formulas and viceversa.

Given a triple pattern $t$ and a mapping $\mu$ such that 
$\var(t) \subseteq \dom(\mu)$, we denote by $\mu(t)$ the triple obtained by replacing the variables in $t$ according to $\mu$. 
Overloading the above definition, we denote by $\mu(P)$ the graph pattern obtained by the recursive substitution of variables in every triple pattern and filter constraint occurring in the graph pattern $P$ according to $\mu$.

\paragraph{Semantics of SPARQL graph patterns.}
 The evaluation of a SPARQL graph pattern $P$ over an RDF graph $G$ is defined as a function $\ev{P}{G}$ (or $\ev{P}{}$ where $G$ is clear from the context) which returns a multiset of solution mappings. 
 Let $P_1,P_2,P_3$ be graph patterns and $C$ be a filter constraint.
The evaluation of a graph pattern $P$ over a graph $G$
is defined recursively as follows:
\begin{enumerate}
\item If $P$ is a triple pattern $t$, then 
$\ev{P}{G} = \{ \mu \mid \dom(\mu) = \var(t) \land \mu(t) \in G \}$
where each mapping $\mu$ has cardinality 1.	
\item $\ev{(P_1 \AAND P_2)}{G} = \ev{P_1}{G} \Join \ev{P_2}{G}$    
\item If $P$ is $(P_1 \OPTIONAL P_2)$ then
 \begin{itemize} 
 \item[(a)] if $P_2$ is $(P_3 \FILTER C)$ then $\ev{P}{G} = \ev{P_1}{G} \lojoin_{f(C)} \ev{P_3}{G}$
 \item[(b)] else $\ev{P}{G} = \ev{P_1}{G} \lojoin_{(\true)} \ev{P_2}{G}$  
 \end{itemize}
\item $\ev{(P_1 \MINUS P_2)}{G} =  \ev{P_1}{G} - \ev{P_2}{G}$  
\item $\ev{( P_1 \NEX P_2 )}{G} = \{ \mu \mid \mu \in \ev{P_1}{G} \land \ev{\mu(P_2)}{G}  = \emptyset \}$ 
\item $\ev{(P_1 \UNION P_2)}{G} =  \ev{P_1}{G} \cup \ev{P_2}{G}$  
\item $\ev{(P_1 \FILTER C)}{G} =  \sigma_{f(C)}(\ev{P_1}{G})$ 
\end{enumerate}

\section{Types of negation in SPARQL}
\label{sec:neg-types}
We can distinguish four types of negation in SPARQL:
negation of filter constraints, 
negation as failure,
negation by $\MINUS$ and
negation by $\NEX$.
The main features of these types of negation will be discussed in this section.     

\subsubsection{Negation of filter constraints.}
\label{sec:not}
The most basic type of negation in SPARQL is the one allowed in filter graph patterns by including constraints of the form $(!C)$.
Following the semantics of SPARQL, a graph pattern $(P \FILTER (!C))$, returns the mappings $\mu$ in $\ev{P}{}$ such that $\mu$ satisfies the filter constraint $(!C)$, i.e. $\mu$ does not satisfy the constraint $C$.

Considering that the evaluation of a graph pattern could return mappings with unbound variables (similar to NULL values in SQL), SPARQL uses a three-valued logic for evaluating filter graph patterns, i.e. the evaluation of a filter constraint $C$ can result in $\true$, $\false$ or $\error$. 
 For instance, given a mapping $\mu$, the constraint $?X = 1$ evaluates to $\true$ when $\mu(?X) = 1$, $\false$ when $\mu(?X) \neq 1$, and $\error$ when $?X \notin \dom(\mu)$. Note that $\FILTER$ eliminates any solutions that result in $\false$ or $\error$.   
 Recalling the semantics defined in Table \ref{table:tvl}, 
we have that $(!C)$ evaluates to $\true$ when $C$ is $\false$, $\false$ when $C$ is $\true$, and $\error$ when $C$ is $\error$.
 This type of negation, called \emph{strong negation} \cite{91281}, allows to deal with incomplete information in a similar way to the NOT operator of SQL.
 
 The negation of filter constraints is a feature well established in SPARQL and does not deserve major discussion. In the rest of the paper we concentrate our interest on the negation of graph patterns.

\subsubsection{Negation as failure.}
\label{sec:diff}
SPARQL 1.0 does not include an operator to express the negation of graph patterns. 
 In an intent of patching this issue, the SPARQL specification remarks that the negation of graph patterns can be implemented as a combination of an optional graph pattern and a filter constraint containing the bound operator (see \cite{10155}, Sec. 11.4.1).
  This style of negation, called \emph{negation as failure} in logic programming, can be illustrated as a graph pattern $P$ of the form
$((P_1 \OPT P_2) \FILTER (! \bound(?X)))$
where $?X$ is a variable of $P_2$ not occurring in $P_1$.
Note that, the evaluation of $P$ returns the mappings of $\ev{(P_1 \OPT P_2)}{}$ 
satisfying that variable $?X$ is unbounded, i.e. $?X$ does not match $P_2$.
In other words, $P$ returns ``the solution mappings of $P_1$ that are not compatible with the solutions mappings of $P_2$''. 
 Unfortunately, this simulation does no work for the general case \cite{91031}. 
 A discussion of the issues and the general solution for implementing negation by failure is included in Appendix \ref{sec:failure}.  
 
 In order to facilitate the study of negation by failure in SPARQL, we introduce the operator $\DIFF$ as an explicit way of expressing it.
 
\begin{definition}[DIFF] 
Let $P_1$ and $P_2$ be graph patterns.
The DIFF operator is defined as 
$\ev{(P_1 \DIFF P_2)}{} = \{ \mu_1 \in \ev{P_1}{} \mid \forall \mu_2 \in \ev{P_2}, \mu_1 \nsim \mu_2 \}$.
\end{definition}

 It is very important to note that the DIFF operator is not defined in SPARQL 1.0 nor in SPARQL 1.1. However, it can be directly implemented with the difference operator of the algebra of solution mappings.
 
\subsubsection{Negation by $\MINUS$.}
 SPARQL 1.1 introduced the MINUS operator as an explicit way of expressing the negation (or difference) of graph patterns. 
 Note that $\DIFF$ and $\MINUS$ have similar definitions.
 The difference is given by the restriction about disjoint mappings included by the  $\MINUS$ operator.
 Such restriction, named Antijoin Restriction inside the SPARQL working group\footnote{\url{http://lists.w3.org/Archives/Public/public-rdf-dawg/2009JulSep/0030.html}}, was introduced to avoid solutions with vacuously compatible mappings.   
 Such restriction causes different results for $\DIFF$ and $\MINUS$.  
 Basically, if $P_1$ and $P_2$ do not have variables in common then 
$\ev{(P_1 \DIFF P_2)}{} = \emptyset$
whereas
$\ev{(P_1 \MINUS P_2)}{} = \ev{P_1}{}$.

Note that both, DIFF and MINUS resemble the EXCEPT operator of SQL\cite{70072}, i.e. given two SQL queries $Q_1$ and $Q_2$, the expression $(Q_1 \EXCEPT Q_2)$ allows to return all rows that are in the table obtained from $Q_1$ \emph{except} those that also appear in the table obtained from $Q_2$.
Considering that $\EXCEPT$ makes reference to the difference of two relations (tables), we can say that $\DIFF$ and $\MINUS$ allow to express the \emph{difference of two graph patterns}.



\subsubsection{Negation by $\NEX$.}
Another type of negation defined in SPARQL 1.1 is given by the $\NEX$ operator. 
The main feature of this type of negation is the possible occurrence of correlation.  
Given a graph pattern $P = (P_1 \NEX P_2)$, we will say that $P_1$ and $P_2$ are \emph{correlated} when $\var(P_1) \cap \var(P_2) \neq \emptyset $, 
i.e. there exist variables occurring in both $P_1$ and $P_2$, and such variables 
are called \emph{correlated variables}.
 In this case, the evaluation of $P$ is attained by replacing variables in $P_2$ with the corresponding values given by the current mapping $\mu$ of $\ev{P_1}{}$, and testing whether the evaluation of the graph pattern $\mu(P_2)$ returns no solutions.
 This way of evaluating correlated queries is based on the nested iteration method \cite{50640} of SQL.
 The correlation of variables in $\NEX$ introduces several issues that have been studied in the context of subqueries in SPARQL \cite{90361,90383}.
 Some of these issues are discussed in Section \ref{sec:nex}.


$\DIFF$, $\MINUS$ and $\NEX$ are three different ways of expressing negation in SPARQL. 
$\DIFF$ and $\MINUS$ represent the difference of two graph patterns whereas $\NEX$ tests the presence of a graph pattern. 

\section{Expressive Power}
\label{sec:consolidation}

Let us recall some definitions related to expressive power. 
By the expressive power of a query language, we understand the set of all queries expressible in that language \cite{70001}.
We will say that an operator $O$ is expressible in a language $L$ iff a subset of the operators of $L$ allow to express the same queries as $O$.  
Finally, a language $L$ contains a language $L'$ iff every operator of $L'$ is expressible by $L$.

 In this section we study the expressive power of the negation operators defined in the above section, i.e. $\DIFF$, $\MINUS$ and $\NEX$. 
 First, we introduce a ``core'' SPARQL algebra which contains the W3C SPARQL algebra, but having the advantage of being smaller and simpler. 
 Such core algebra is the basis to define equivalences among graph pattern operators, in particular those implementing negation. 

\subsection{The core SPARQL algebra}
\label{sec:core}
Let us introduce a ``core'' algebra for SPARQL which is able to express all the high-level operators of the W3C SPARQL language. 
Recall that the W3C SPARQL algebra, defined in Section \ref{sec:patterns}, is composed by the operators of \emph{projection}, \emph{selection}, \emph{join}, \emph{union}, \emph{left-join} and \emph{minus}. 
Our core algebra is based on a new operator called ``simple difference''.

\begin{definition}[Simple difference]
 The simple difference between two solution mappings, $\Omega_1$ and $\Omega_2$, is defined as 
$\Omega_1 \setminus \Omega_2 = \{ \mu_1 \in \Omega_1 \mid \forall \mu_2 \in \Omega_2, \mu_1 \nsim \mu_2 \}$
where  
$\cardi{\mu_1}{\Omega_1 \setminus \Omega_2} = \cardi{\mu_1}{\Omega_1}$.
\end{definition}

\begin{definition}[Core SPARQL algebra]
\label{def:core}
The core SPARQL algebra is composed by the operations of projection, selection, join, union and simple difference. 
\end{definition}

Next, we will show that the core SPARQL algebra contains the W3C SPARQL algebra. Specifically, we will show that the operators of \emph{difference}, \emph{left-join} and \emph{minus} (of the W3C SPARQL algebra) can be simulated with the \emph{simple difference} operator (of the core SPARQL algebra). 

\begin{lemma}
\label{lemma:diff-sdiff}
The difference operator (of the W3C SPARQL algebra) is expressible in the core SPARQL algebra.
\end{lemma}
\begin{proof}
Recall that the \emph{difference} operator is defined as\\
$\Omega_1 \setminus_F \Omega_2 = \{ \mu_1 \in \Omega_1 \mid \forall \mu_2 \in \Omega_2, (\mu_1 \nsim \mu_2) \lor (\mu_1 \sim \mu_2 ~\land~ (\mu_1 \cup \mu_2)(F) = \false ) \}$.
\\
Assume that $\Omega_1' = \Omega_1 \setminus (\Omega_1 \setminus_F \Omega_2)$, i.e. $\Omega_1'$ is the complement of $\Omega_1 \setminus_F \Omega_2$.
We have that $\mu_1 \in \Omega_1'$ 
iff there exists a mapping $\mu_2 \in \Omega_2$ satisfying that $\mu_1 \sim \mu_2$ and $(\mu_1 \cup \mu_2)(F)$ evaluates to either $\true$ or $\error$.
Hence, we will have that  
\[
\Omega_1 \setminus_F \Omega_2 = \Omega_1 \setminus ( \sigma_F(\Omega_1 \Join \Omega_2) \cup ((\Omega_1 \Join \Omega_2) \setminus (\sigma_{(F \lor \neg F)}(\Omega_1 \Join \Omega_2))))
\]
where $\sigma_F(\Omega_1 \Join \Omega_2)$ corresponds to the mappings evaluating $F$ as $\true$, and $((\Omega_1 \Join \Omega_2) \setminus (\sigma_{(F \lor \neg F)}(\Omega_1 \Join \Omega_2)))$ corresponds to the mappings evaluating $F$ as $\error$.

Note that the above equivalence is based on a faithful interpretation as occurs in the current specification.
However, the document of errata in SPARQL 1.1 includes a clarification that differs of our interpretation\footnote{\url{https://www.w3.org/2013/sparql-errata#errata-query-12}}:          
``The definition of Diff should more clearly say that expressions that evaluate with an error are considered to be false.''
Under this interpretation, the proof is much simpler\footnote{In an earlier version of this report, the proof of Lemma \ref{lemma:diff-sdiff} was based on the full definition of LeftJoin presented in Sec. 18.5 of the current SPARQL 1.1 specification, a statement that was reported to be wrong in the errata \url{http://www.w3.org/2013/sparql-errata#errata-query-7a}
(We thank to R. Kontchakov, who called to our attention this issue
in a personal communication, May 19th, 2016.)
 The same bug in the proof of Lemma 1 was included in
the article ``Negation in SPARQL'' (Alberto Mendelzon International Workshop on Foundations of Data Management, AMW'2016).
Fortunately, the problem was with the proof, not with the statement
of the Lemma which remains untouched and thus does not affect the rest
of results in such article nor in this report.}:
\[ 
\Omega_1 \setminus_F \Omega_2 = \Omega_1 \setminus ( \sigma_F(\Omega_1 \Join \Omega_2)).
\]
\end{proof}


\begin{lemma}
\label{lemma:ljoin-sdiff}
The left-join operator (of the W3C SPARQL algebra) is expressible in the core SPARQL algebra.
\end{lemma}
\begin{proof}
Recall that $\Omega_1 \lojoin_F \Omega_2 = \sigma_F(\Omega_1 \Join \Omega_2) \cup (\Omega_1 \setminus_F \Omega_2)$.
Considering that the \emph{difference} expression $(\Omega_1 \setminus_F \Omega_2)$ can be expressed with \emph{simple difference} (Lemma \ref{lemma:diff-sdiff}), we have that
\[
\Omega_1 \lojoin_F \Omega_2 = \sigma_F(\Omega_1 \Join \Omega_2)~\cup~(\Omega_1 \setminus ( \sigma_F(\Omega_1 \Join \Omega_2)))
\]
\end{proof}

\begin{lemma}
\label{lemma:sdiff:minus}
The minus operator (of the W3C SPARQL algebra) is expressible in the core SPARQL algebra.
\end{lemma}
\begin{proof}
Note that the difference between $\Omega_1 - \Omega_2$ and $\Omega_1 \setminus \Omega_2$ is given by the condition $\dom(\mu_1) \cap \dom(\mu_2) = \emptyset$.
 Intuitively, it applies that $(\Omega_1 - \Omega_2) = \Omega_1$ when $\dom(\Omega_1) \cap \dom(\Omega_2) = \emptyset$. 
In the case that $\dom(\Omega_1) \cap \dom(\Omega_2) \neq \emptyset$, the rewriting is a little more tricky.

Recall that the operation $\Omega_1 - \Omega_2$ is defined by the set\\
$S = \{ \mu_1 \in \Omega_1 \mid \forall \mu_2 \in \Omega_2, \mu_1 \nsim \mu_2 \lor \dom(\mu_1) \cap \dom(\mu_2) = \emptyset \}$. \\
Note that the complement of $S$ in $\Omega_1$ is the set\\
$S^c = \{ \mu_1 \in \Omega_1 \mid \exists \mu_2 \in \Omega_2, \mu_1 \sim \mu_2 \land \dom(\mu_1) \cap \dom(\mu_2) \neq \emptyset \}$.\\
Assuming that $\dom(\Omega_1) \cap \dom(\Omega_2) = \{ ?X_1, \dots, ?X_n \}$, 
$(\Omega_1 - \Omega_2)$ can be simulated with an expression of the form
$\Omega_1 \setminus ( \sigma_F (\Omega_1 \Join \Omega_3))$
where  
$\Omega_3$ is a copy of $\Omega_2$ where every variable $X_i$ has been replaced with a free variable $X_i'$, and 
$F$ is a selection formula of the form 
$(\dots(?X_1 = ?X_1' \land ?X_2 = ?X_2') \dots ) \land ?X_n = ?X_n')$.
 Note that, the expression $\sigma_F (\Omega_1 \Join \Omega_3)$ returns the set $S^c$, which is subtracted from $\Omega_1$ to obtain $S$. 
It is easy to see that the proposed equivalent expression preserves the cardinalities. 
Hence, we have proved that $(\Omega_1 - \Omega_2)$ can be simulated with $(\Omega_1 \setminus \Omega_2)$.   
\end{proof}

Based on Lemmas \ref{lemma:diff-sdiff}, \ref{lemma:ljoin-sdiff} and \ref{lemma:sdiff:minus}, we can present our main result about the expressive power of the core SPARQL algebra.

\begin{theorem}
\label{theo:saeqma}
The core SPARQL algebra is equivalent with the W3C SPARQL algebra. 
\end{theorem}

This result implies that the W3C SPARQL algebra could be implemented by a subset of the original operators (projection, selection, join and union), plus the simple difference operator.

\subsection{$\spqldiff$: A core fragment with negation}
Let us redefine the $\DIFF$ operator by using the simple difference operator of the core algebra as follows: 
\[
\ev{(P_1 \DIFF P_2)}{} = \ev{P_1}{} \setminus \ev{P_2}{}.
\] 
Assume that $\spqldiff$ is the language defined (recursively) by graph patterns of the form $(P \AAND P')$, $(P \UNION P')$, $(P \DIFF P')$ and $(P \FILTER C)$.
 Next, we will show that $\spqldiff$ can express a fragment of SPARQL whose $\NEX$ graph patterns satisfy a restricted notion of correlated variables.     

\subsubsection{Expressing OPTIONAL and MINUS with DIFF}
Here, we show that the $\OPT$ and $\MINUS$ operators can be expressed in $\spqldiff$ by using the $\DIFF$ operator.


\begin{lemma}
\label{lemma:OPT-DIFF}
The $\OPTIONAL$ operator is expressible in $\spqldiff$.
\end{lemma}
\begin{proof}
Recall the semantics of optional graph patterns as presented in Section \ref{sec:patterns}.
Given a graph pattern $P$ of the form $(P_1 \OPTIONAL P_2)$, we have two cases:\\
(i) if $P_2$ is $(P_3 \FILTER C)$ then $\ev{P}{}{} = \ev{P_1}{}{} \lojoin_{f(C)} \ev{P_3}{}{}$;\\
(ii) else, $\ev{P}{}{} = \ev{P_1}{}{} \lojoin_{(\true)} \ev{P_2}{}{}$.\\ 
Note that this definition follows an operational semantics in the sense that the evaluation of $(P_1 \OPTIONAL P_2)$ depends on the structure of $P_2$, i.e. when $P_2$ is a filter graph pattern applies case (i), otherwise case (ii) . Let us analyze both cases.

Case (i): 
We have that 
$\ev{P_1}{}{} \lojoin_{f(C)} \ev{P_3}{}{} = \sigma_{F}(\Omega_1 \Join \Omega_3) \cup (\Omega_1 \setminus_{F} \Omega_3)$ 
where 
$\Omega_1 = \ev{P_1}{}{}$, $\Omega_3 = \ev{P_3}{}{}$ and $F = f(C)$. 
If we rewrite the right side in terms of simple difference (Lemma \ref{lemma:ljoin-sdiff}),
we obtain   
$ 
\sigma_{F}(\Omega_1 \Join \Omega_3) 
\cup 
(\Omega_1 \setminus ( \sigma_F(\Omega_1 \Join \Omega_3)))
$.
Similarly, the graph pattern $(P_1 \OPTIONAL (P_3 \FILTER C))$
can be rewritten as    
\[
(((P_1 \AAND P_3) \FILTER C ) \UNION (P_1 \DIFF ((P_1 \AAND P_3) \FILTER C))).
\]


Case (ii):
We have that $\ev{P_1}{}{} \lojoin_{(\true)} \ev{P_2}{}{} = \sigma_{(\true)}(\Omega_1 \Join \Omega_2) \cup (\Omega_1 \setminus_{(\true)} \Omega_2)$ where $\Omega_1 = \ev{P_1}{}{}$ and $\Omega_2 = \ev{P_2}{}{}$. 
If we rewrite the right side in terms of simple difference (Lemma \ref{lemma:ljoin-sdiff}), we obtain
$ 
\sigma_{(\true)}(\Omega_1 \Join \Omega_2) 
\cup 
(\Omega_1 \setminus ( \sigma_{(\true)}(\Omega_1 \Join \Omega_2)))
$, which can be reduced to
$(\Omega_1 \Join \Omega_2) \cup (\Omega_1 \setminus ( \Omega_1 \Join \Omega_2))$.  
Similarly, the graph pattern $(P_1 \OPTIONAL P_2)$ can be rewritten as  
\[
((P_1 \AAND P_2) \UNION (P_1 \DIFF (P_1 \AAND P_2))).
\]

Hence, we can use the above transformation to simulate the operational semantics of $\OPT$ by using $\DIFF$, i.e. $\OPT$ is expressible in $\spqldiff$. 
\end{proof}

\begin{lemma}
\label{lemma:MINUS-DIFF}
The $\MINUS$ operator is expressible in $\spqldiff$.
\end{lemma}
\begin{proof}
Given a graph pattern $(P_1 \MINUS P_2)$, we have that:\\
(i) If $\var(P_1) \cap \var(P_2) = \emptyset$ then
$\ev{(P_1 \MINUS P_2)}{} = \ev{P_1}{}$.\\
(ii) Else, $\ev{(P_1 \MINUS P_2)}{} = \Omega_1 - \Omega_2$ 
where $\Omega_1 = \ev{P_1}{}$ and $\Omega_2 = \ev{P_2}{}$.
Recalling the simulation of $(\Omega_1 - \Omega_2)$ with $(\Omega_1 \setminus \Omega_2)$ presented in Lemma \ref{lemma:sdiff:minus}, we have that $(P_1 \MINUS P_2)$ can be rewritten to a graph pattern of the form
$(P_1 \DIFF ((P_1 \AAND P_3) \FILTER C))$ where $P_3$ and $C$ must be created by following a procedure similar to the one described in Lemma \ref{lemma:sdiff:minus}.
Hence, we have proved that MINUS can be expressed in $\spqldiff$. 
\end{proof}

\subsubsection{NOT-EXISTS vs DIFF}
\label{sec:nex}
 In order to give a basic idea of the expressive power of NOT-EXISTS, we will  analyze its relationship with the DIFF operator. 
 Given the graph patterns $P = (P_1 \NEX P_2)$ and $P' = (P_1 \DIFF P_2)$,
the basic case of equivalence between $P$ and $P'$ is given when $\var(P_1) \cap \var(P_2) = \emptyset$.
 Note that, the non occurrence of correlated variables between $P_1$ and $P_2$ implies that, both $P$ and $P'$ are restricted to test whether $P_2$ returns an empty solution.

Intuitively, one can assume that any graph pattern $P = (P_1 \NEX P_2$) can be directly translated into $P' = (P_1 \DIFF P_2)$ such that $P$ and $P'$ are equivalent.
This is for example argued by Kaminski et. al \cite{91302} (Lemma 3). 
However, the translation given there does not work.
For instance, consider the graph $G$ = \{(a,p,b),(f,p,b),(c,q,d),(e,r,a)\} and the graph patterns\\
$P$ = (($?X$ p b) $\NEX$ ( ($?Z$ q d) $\NEX$ ($?W$ r $?X$))) and \\
$P'$ = (($?X$ p b) $\DIFF$ ( ($?Z$ q d) $\DIFF$ ($?W$ r $?X$))).\\
In this case, $P$ are $P'$ are not equivalent such that 
$\ev{P}{G} = \{ \mu_{?X \to a} \}$ whereas   
$\ev{P'}{G} = \{ \mu_{?X \to a}, \mu'_{?X \to f} \}$.

Despite the above negative result, it will be interesting to identify a subset of $\NEX$ graph patterns that can be expressed by using the $\DIFF$ operator. 
 To do this, we will introduce the notion of ``safe'' and ``unsafe'' variables.
 The set of \emph{safe variables} in a graph pattern $P$, denoted $\svar(P)$, is defined recursively as follows: 
 If $P$ is a triple pattern, then $\svar(P) = \var(P)$;
 If $P$ is $(P_1 \AAND P_2)$ then $\svar(P) = \svar(P_1) \cup \svar(P_2)$;
 If $P$ is $(P_1 \UNION P_2)$ or $(P_1 \OPT P_2)$ then $\svar(P) = \svar(P_1) \cap \svar(P_2)$;
 If $P$ is $(P_1 \FILTER C)$, $(P_1 \MINUS P_2)$, $(P_1 \NEX P_2)$ or $(P_1 \DIFF P_2)$ then $\svar(P) = \svar(P_1)$.
 Therefore, a variable occurring in $\svar(P)$ is called a \emph{safe variable}, otherwise it is considered \emph{unsafe} in $P$.

\begin{definition}[$\spqlexs$]
\label{def:spqlex}
Define $\spqlexs$ as the fragment of SPARQL graph patterns satisfying that the occurrence of a subpattern $(P \NEX P')$ implies that, for every correlated variable $?X$ between $P$ and $P'$, it holds that $?X \in \svar(P')$.
\end{definition}

Note that, $\spqlex$ does not allow graph patterns of the form \\
$(P_1 \NEX (P_2 \NEX P_3))$ where $P_3$ contains correlated variables occurring in $P_1$ but not occurring in $P_2$.
 
\begin{lemma}
\label{lemma:NEX2-DIFF}
The $\NEX$ graph patterns allowed in $\spqlexs$ are expressible in $\spqldiff$. 
\end{lemma}
\begin{proof}
Following the definition of $\spqlexs$, we have that for any graph pattern $(P_1 \NEX P_2)$ it satisfies that the domain of $\ev{P_2}{}$ contains all the correlated variables between $P_1$ and $P_2$. Given such condition, it is easy to see that $\ev{(P_1 \DIFF P_2)}{}$ returns the same solutions as $\ev{(P_1 \NEX P_2)}{}$. 
\end{proof}

Based on Lemmas \ref{lemma:OPT-DIFF}, \ref{lemma:MINUS-DIFF} and \ref{lemma:NEX2-DIFF}, we can present our main result about the expressive power of the $\DIFF$ operator and $\spqldiff$.    

\begin{theorem}
\label{theo:saeqma}
$\spqldiff$ contains $\spqlexs$.
\end{theorem}


\section{Properties of the SPARQL negation operators}
 In this section we evaluate the negation operators in terms of elementary equivalences found in set theory.
Specifically, we consider the following axioms concerning set-theoretic differences \cite{70116}:

\smallskip

\begin{tabular}{llllll}
\textbf{(a)} & $A \setminus A \equiv \emptyset$ & \verb+   + & \textbf{(g)} & $A \setminus (A \cap B) \equiv A \setminus B$ \\
\textbf{(b)} & $A \setminus \emptyset \equiv A$ & \verb+   + & \textbf{(h)} & $A \cap (A \setminus B) \equiv A \setminus B$ \\
\textbf{(c)} & $\emptyset \setminus A \equiv \emptyset$ & \verb+   + & \textbf{(i)} & $(A \setminus B) \cup B \equiv A \cup B$ \\
\textbf{(d)} & $A \setminus (A \setminus (A \setminus B)) \equiv A \setminus B$ & \verb+   + & \textbf{(j)} & $(A \cup B) \setminus B \equiv A \setminus B$ \\
\textbf{(e)} & $(A \cap B) \setminus B \equiv \emptyset$ & \verb+   + & \textbf{(k)} & $A \setminus (B \cap C) \equiv (A \setminus B) \cup (A \setminus C)$ \\
\textbf{(f)} & $(A \setminus B) \cap B \equiv \emptyset$ & \verb+   + & \textbf{(l)} & $A \setminus (B \cup C) \equiv (A \setminus B) \cap (A \setminus C)$ \\
\end{tabular}

\smallskip

 Our motivation to consider these equivalences is given by the intrinsic nature of expressing negation in SPARQL, i.e. the difference of graph patterns.
 Moreover, set-theoretic equivalences are supported by other database languages like relational algebra and SQL.    
 Recalling that SPARQL works under bag semantics, we will evaluate the behavior of the negation operators by examining both, set and bag semantics.
 
 Let us define the notion of equivalence between graph patterns. 
Two graph patterns $P_1$ and $P_2$ are equivalent, denoted by $P_1 \equiv P_2$, if it satisfies that $\ev{P_1}{G} = \ev{P_2}{G}$ for every RDF graph $G$.
 In order to evaluate the set-based equivalences in the context of SPARQL, we need to assume two conditions:
\begin{itemize}
\item[(1)]  As a general rule, a set-based operator requires two ``objects'' with the same structure (e.g. two tables with the same schema). Such requirement is implicitly satisfied in SPARQL thanks to the definition of solution mappings as partial functions. 
Basically, if we fix the set of variables $V$, then we have that for every pair of solution mappings $\Omega_1$ and $\Omega_2$ it satisfies that $\dom(\Omega_1) = \dom(\Omega_2)$, i.e. $\Omega_1$ and $\Omega_2$ have the same domain of variables (even when some variables could be unbounded). 
 Hence, set-based operations like union and difference can be applied to the algebra of solution mappings.
 
\item[(2)] SPARQL does not provide an explicit operator for intersecting two graph patterns $P_1$ and $P_2$, i.e. the solution mappings which belong both to $\ev{P_1}{}$ and $\ev{P_2}{}$. 
Note however, that a graph pattern $(P_1 \AAND P_2)$ resembles the intersection operation (under set-semantics) when $\ev{P_1}{}$ and $\ev{P_2}{}$ have the same domain of variables. Recalling condition (1), we will assume that the AND operator could be used to play the role of the intersection operator in the set-theoretic equivalences defined above.  
\end{itemize}

Given the above two conditions, we can apply a direct translation from a set-theoretic equivalence to a graph pattern equivalence. 
Specifically, the set-difference operator will be mapped to a SPARQL negation operator (DIFF, MINUS or NOT-EXISTS), the set-intersection operator will be mapped to AND, and the set-union operator will be replaced by UNION.
 Considering the specific features of NOT-EXISTS (i.e. correlation of variables), we will restrict our analysis to DIFF and MINUS. 

Let $P_1$, $P_2$, $P_3$, $P_4$ and $P_5$ be graph patterns satisfying that 
$\ev{P_1}{} = \emptyset$,
$\ev{P_2}{} = \{ \mu_0 \}$,
$\var(P_3) = \var(P_4)$
and
$\var(P_3) \cap \var(P_5) = \emptyset$. 
 By combining the above four graph patterns, we conducted a case-by-case analysis consisting of twenty five cases for equivalences (a)-(j) and one hundred twenty five cases for equivalences (k) and (l). 
 The differences between set and bag semantics were specially considered for equivalences (h), (i), (k) and (l). 
 Next, we present our findings.

DIFF satisfies most equivalences with exception of (h), (i), (k) and (l).
Equivalence (h) presents five cases which are not valid under bag semantics, although they are valid under set semantics. A similar condition occurs with ten and seven cases for equivalences (k) and (l) respectively. 
Additionally, we found ten cases which are not satisfied by equivalence (i). 
 
MINUS does not satisfy equivalences (e) to (l).
We found several cases where equivalences (e), (f), (g), (j) and (k) are not satisfied.
Similarly, there exist multiple cases where equivalences (h), (i) and (l) do not apply under bag semantics, but works for set semantics.
 We would like to remark that the ``odd results'' presented by the MINUS operator arise because the restriction about disjoint solution mappings introduced in its definition.

In summary, we have that each negation operator presents a particular behavior for the axioms studied here. 
Although none of them was able to satisfy all the axioms, we think that it does not mean that they are badly defined. 
In fact, the heterogeneity of the operators is a motivation to study their intrinsic properties and to try the definition of a set of desired properties for negation in SPARQL. 
The details about our case-by-case analysis are included in the Appendix.

\section{Conclusions}
In this paper we presented a systematic analysis of the types of negation supported in SPARQL 1.0 and SPARQL 1.1. 
After introducing the standard relational negation (the DIFF operator) we were
able to build a core and intuitive algebra (the same as the standard relational algebra) in SPARQL and prove that it is able to define the graph pattern operators.

We think that having a clear understanding of the operators of a language
(in this case, the operators of negation of SPARQL) helps both, developers
of databases, and users of the query language.
 We also think that the core language we identified (which is precisely the well
 known and intuitive relational algebra) is a much easier way to
 express queries for database practitioners, who learn from the beginning
 SQL, which now with this new algebra, can be found in the world of SPARQL.

\smallskip
\noindent
\textbf{Acknowledgements.}
R. Angles and C. Gutierrez are founded by the Millennium Nucleus Center
for Semantic Web Research under Grant NC120004.

\bibliographystyle{splncs}

\bibliography{referencias}

\appendix

\section{Simulation of Negation as Failure in SPARQL}
\label{sec:failure}

In order to conduct our discussion about the simulation of negation by failure in SPARQL, we need to introduce two concepts: RDF Datasets and the $\GRAPH$ operator. 

An \emph{RDF dataset} $D$ is a set 
$\{G_0,\langle u_1,G_1 \rangle,\dots,\langle u_n,G_n \rangle\}$
where each $G_i$ is a graph and each $u_j$ is an IRI.
$G_0$ is called the \emph{default graph} of $D$ and it is denoted $\dg(D)$.
Each pair $\langle u_i,G_i \rangle$ is called a \emph{named graph}; 
define $\name(G_i)_D = u_i$ and $\graph(u_i)_D = G_i$.
The set of IRIs $\{u_1,\dots,u_n\}$ is denoted $\names(D)$.
Every dataset satisfies that:
(i) it always contains one default graph (which could be empty);
(ii) there may be no named graphs; 
(iii) each $u_j$ is distinct; and
(iv) $\blank(G_i) \cap \blank(G_j) = \emptyset$ for $i \neq j$.
Finally, the \emph{active graph} of $D$ is the graph $G_j$ used for querying $D$.

Next, we also extend the graph pattern evaluation function. 
The evaluation of a graph pattern $P$ over a dataset $D$ with active graph $G$ will be denoted as $\eval{P}{D}{G}$ (or $\eval{P}{}{G}$ where $D$ is clear from the context).
Specifically, $\eval{P}{D}{G}$ returns the multiset of mappings that matches the dataset $D$ according to the graph pattern $P$.    

Let us extend the (recursive) definition of graph pattern.
Given a graph pattern $P$ and $n \in I \cup V$, we have that $(n \GRAPH P)$ is also a graph pattern.
The semantics of $(n \GRAPH P)$ is given as follows:
\begin{enumerate}
\item If $u \in I$ then $\eval{(u \GRAPH P)}{D}{}$ = $\eval{P}{D}{G}$ where $G = \graph(u)_D$ \\
\item If $?X \in V$ then $\eval{(?X \GRAPH P)}{D}{}$ =  $\bigcup_{v \in \names(D)} ( \eval{P}{D}{\graph(v)_D} \Join \{ \mu_{?X \to v} \})$  
\end{enumerate}

\smallskip

We are ready to begin our analysis.
In Section \ref{sec:diff}, we argued that the Negation by Failure, and also $(P_1 \DIFF P_2)$, can be simulated by a pattern of the form 
 \begin{equation}\label{eq:neg1}
((P_1 \OPT P_2) \FILTER (! \bound(?X)))
 \end{equation}
where $?X$ is a variable of $P_2$ not occurring in $P_1$.
Unfortunately, there are two issues with this solution:
\begin{itemize}
\item Variable $?X$ cannot be an arbitrary variable.
For example, $P_2$ could be in turn an optional pattern $(P_3 \OPT P_4)$ where only
variables in $P_3$ are relevant to evaluate the condition $(! \bound(?X))$. 
\item If $P_1$ and $P_2$ do not have variables in common, then there is no variable $?X$ to check unboundedness.    
\end{itemize}

Angles and Gutierrez \cite{90006} identified these issues and
proposed a solution for them. The solution was shown to have a small bug.

\paragraph{Perez's solution.}
A second approach which fixed the previous bag was presented in \cite{90851}, where it was proposed that $(P_1 \DIFF P_2)$ can be simulated by the pattern
\begin{equation}\label{eq:neg2}
( (P_1 \OPT~(P_2 \AAND~(?X_1~?X_2~?X_3))) \FILTER \neg \bound (?X_1))
\end{equation}
where $?X_1,?X_2,?X_3$ are fresh variables mentioned neither in $P_1$ nor in $P_2$.

This solution works well when the empty graph pattern is not allowed
 in the grammar. In the presence of it,
the counter example is given when $P_1$ and $P_2$ are empty graph patterns and the default graph is the empty graph.
In this case, we have that 
$\eval{(P_1 \DIFF P_2)}{D}{G_0} = \emptyset$ 
whereas 
$\eval{(\ref{eq:neg2})}{D}{G_0} = \{ \Omega_0 \}$ .

\paragraph{Polleres's solution.}
A third proposal was sketched by Axel Polleres in the mailing list of the SPARQL W3C Working Group, and discussed with the authors of this paper during the year 2009.
The proposed solution is defined formally in the following proposition.

\begin{proposition}\label{prop:neg}
Let $P_1,P_2$ be graph patterns.
We have that $(P_1 \DIFF P_2)$ is equivalent to 
\begin{equation}\label{eq:neg3}
(((P_1 \OPT P_2) \AAND (g \GRAPH~(?X \text{:p :o}))) \FILTER \neg \bound(?X))
\end{equation}
where $g$ is a named graph containing the single triple (:s :p :o)
and $?X$ is a free variable.
\end{proposition}
\begin{proof}
In order to prove that the above equivalence holds, we have analyzed the corner cases for $(P_1 \DIFF P_2)$.
In Table \ref{tab:neg}, we compare the evaluation of the graph patterns presented in Proposition \ref{prop:neg}, considering all the possible solutions for $\ev{P_1}{}$ and $\ev{P_2}{}$, and including the special case when the default graph is the empty graph.   
We showed that the evaluation of both graph patterns are equivalent for all the cases.
Additionally, Table \ref{tab:neg} includes the results for the graph pattern presented in Equation \ref{eq:neg2}, and shows that it fails in the case (5) when the default graph is the empty graph as described before.  
\end{proof}

\begin{table}[t!]
\centering
\begin{tabular}{|c|c|c|c|c|c|c|c|}
\hline 
  &  &  &  & \multicolumn{2}{|c|}{$G_0 \neq \emptyset$} &  \multicolumn{2}{|c|}{$G_0 = \emptyset$}
\\ \hline

  & $\ev{P_1}{}$ & $\ev{P_2}{}$ & $\ev{(P_1 \DIFF P_2)}{}$ & $\ev{P_3}{}$ & $\ev{P_4}{}$ & $\ev{P_3}{}$ & $\ev{P_4}{}$
\\ \hline 

1 & $\emptyset$ & $\emptyset$ & $\emptyset$ & $\emptyset$ & $\emptyset$ & $\emptyset$ & $\emptyset$
\\ \hline  

2 & $\emptyset$ & $\Omega_0$ & $\emptyset$ & $\emptyset$ & $\emptyset$ & $\emptyset$ & $\emptyset$ 
\\ \hline  

3 & $\emptyset$ & $\Omega_{2}$ & $\emptyset$ & $\emptyset$ & $\emptyset$ & $\emptyset$ & $\emptyset$
\\ \hline  
 
4 & $\Omega_0$ & $\emptyset$ &  $\Omega_0$ & $\Omega_0$ & $\Omega_0$ & $\Omega_0$ & $\Omega_0$
 \\ \hline 

5 & $\Omega_0$ & $\Omega_0$ & $\emptyset$ & $\emptyset$ &  $\emptyset$ & $\Omega_0$ & $\emptyset$ 
\\ \hline  

6 & $\Omega_0$ & $\Omega_{2}$ & $\emptyset$ & $\emptyset$ & $\emptyset$ & -- & --  
\\ \hline  

7 & $\Omega_{1}$ & $\emptyset$ & $\Omega_{1}$ & $\Omega_{1}$ & $\Omega_{1}$ & -- & --
\\ \hline  

8 & $\Omega_{1}$ & $\Omega_0$ & $\emptyset$ & $\emptyset$ & $\emptyset$ & -- & -- 
\\ \hline  

9 & $\Omega_{1}$ & $\Omega_{1}$ & $\emptyset$ & $\emptyset$ & $\emptyset$ & -- & --
\\ \hline  

10 & $\Omega_{1}$ & $\Omega_{2}$ & $\Omega_{1} \setminus \Omega_{2}$ & $\Omega_{1} \setminus \Omega_{2}$ & $\Omega_{1} \setminus \Omega_{2}$ & -- & --
\\ \hline  

11 & $\Omega_{1}$ & $\Omega_{3}$ & $\emptyset$ & $\emptyset$ & $\emptyset$ & -- & --
\\ \hline  

\end{tabular}
\smallskip
\caption{Comparison of two implementations of difference of SPARQL graph patterns.
Assume that $\Omega_0 = \{  \mu_0 \}$ (join identity), 
and $\Omega_1,\Omega_2,\Omega_3$ are sets of mappings 
distinct of $\Omega_0$.
Additionally,  
$\dom(\Omega_1) \cap \dom(\Omega_2) \neq \emptyset$ 
and 
$\dom(\Omega_1) \cap \dom(\Omega_3) = \emptyset$.  
$P_3$ and $P_4$ are the graph patterns presented in Equations \ref{eq:neg2} and \ref{eq:neg3} respectively.
Note that, $\ev{P_4}{}$ is equivalent to $\ev{(P_1 \DIFF P_2)}{}$ in all the cases, whereas $\ev{P_3}{}$ fails when $G_0 = \emptyset$ (i.e. the default graph $G_0$ is the empty graph).}
\label{tab:neg}
\end{table}

\section{Case-by-case analysis}

Let $P_{\emptyset}$, $P_1$, $P_2$ and $P_3$ be graph patterns satisfying that 
$\ev{P_{\emptyset}}{} = \emptyset$,
$\ev{P_1}{} \neq \emptyset$,
$\ev{P_2}{} \neq \emptyset$
and
$\ev{P_3}{} \neq \emptyset$.
Also assume that $\Omega_1$, $\Omega_2$ and $\Omega_3$ are solutions mappings satisfying that 
$\dom(\Omega_1) \cap \dom(\Omega_2) \neq \emptyset$, 
$\dom(\Omega_1) \cap \dom(\Omega_3) = \emptyset$ and
$\dom(\Omega_2) \cap \dom(\Omega_2) = \emptyset$
Finally, recall that $\Omega_0$ denotes the multiset consisting of exactly the empty mapping $\mu_0$.

We will show, for each axiom, examples of cases not satisfied by the negation operators. 
We will use a generic operator $\NEG$ to denote any of the operators used in our analysis, i.e. $\DIFF$ and $\MINUS$.

\subsubsection{Axiom (a): $(P_1 \NEG P_2) \equiv \emptyset$ }. 

Both, $\DIFF$ and $\MINUS$ satisfy in all the cases.

\subsubsection{Axiom (b): $(P_1 \NEG P_{\emptyset}) \equiv P_1$ }. 

Both, $\DIFF$ and $\MINUS$ satisfy in all the cases.

\subsubsection{Axiom (c): $(P_{\emptyset} \NEG P_1) \equiv \emptyset$ }. 

Both, $\DIFF$ and $\MINUS$ satisfy in all the cases.

\subsubsection{Axiom (d): $(P_1 \NEG (P_1 \NEG (P_1 \NEG P_2))) \equiv (P_1 \NEG P_2)$ }. 

Both, $\DIFF$ and $\MINUS$ satisfy in all the cases.

\subsubsection{Axiom (e): $((P_1 \AAND P_2) \NEG P_2) \equiv P_{\emptyset}$ }.
\begin{itemize}
\item $\DIFF$ satisfies in all the cases.
\item $\MINUS$ fails in four cases. 
For instance, if $\ev{P_1}{} = \Omega_0$ and $\ev{P_2}{} = \Omega_0$ 
then $\ev{((P_1 \AAND P_2) \MINUS P_2)}{} = \Omega_0$ instead of $\emptyset$. 
\end{itemize}

\subsubsection{Axiom (f): $((P_1 \NEG P_2) \AAND P_2) \equiv P_{\emptyset}$ }.
\begin{itemize}
\item $\DIFF$ satisfies in all the cases.
\item $\MINUS$ fails in eleven cases. 
For instance, if $\ev{P_1}{} = \Omega_0$ and $\ev{P_2}{} = \Omega_2$ then 
$\ev{((P_1 \MINUS P_2) \AAND P_2)}{} = \Omega_2$ instead of $\emptyset$. 
\end{itemize}

\subsubsection{Axiom (g): $(P_1 \NEG (P_1 \AAND P_2)) \equiv (P_1 \NEG P_2)$ }.
\begin{itemize}
\item $\DIFF$ satisfies in all the cases.
\item $\MINUS$ fails in seven cases. 
For instance, if $\ev{P_1}{} = \Omega_1$ and $\ev{P_2}{} = \Omega_0$ 
then 
$\ev{(P_1 \MINUS (P_1 \AAND P_2))}{} = \Omega_0$
whereas 
$\ev{(P_1 \MINUS P_2)}{} = \Omega_1$. 
\end{itemize}

\subsubsection{Axiom (h): $(P_1 \AAND (P_1 \NEG P_2)) \equiv (P_1 \NEG P_2)$ }.
\begin{itemize}
\item Both, $\DIFF$ and $\MINUS$ fail under bag semantics, in five and twelve cases respectively. 
\item For instance, both operators fail when 
$\ev{P_1}{} = \Omega_1$ and $\ev{P_2}{} = \emptyset$
such that  
$\ev{(P_1 \AAND (P_1 \NEG P_2))}{} = \Omega_1 \Join \Omega_1$
whereas  
$\ev{(P_1 \NEG P_2)}{} = \Omega_1$.  
\end{itemize}

\subsubsection{Axiom (i): $((P_1 \NEG P_2) \UNION P_2) \equiv (P_1 \UNION P_2)$ }.
\begin{itemize}
\item $\DIFF$ generates distinct solutions in ten cases, and fails under bag semantics in four cases. 
For instance, if $\ev{P_1}{} = \Omega_0$ and $\ev{P_2}{} = \Omega_1$ then\\ 
$((P_1 \DIFF P_2) \UNION P_2) = \Omega_1$
whereas
$(P_1 \UNION P_2) = \Omega_0 \cup \Omega_1$.
\item $\MINUS$ fails under bag semantics in three cases.
For instance, if $\ev{P_1}{} = \Omega_1$ and $\ev{P_2}{} = \Omega_1$ then
$\ev{((P_1 \MINUS P_2) \UNION P_2)}{} = \Omega_1$
whereas\\
$\ev{(P_1 \UNION P_2)}{} = \Omega_1 \cup \Omega_1$.
\end{itemize}

\subsubsection{Axiom (j): $((P_1 \UNION P_2) \NEG P_2) \equiv (P_1 \NEG P_2)$ }.
\begin{itemize}
\item $\DIFF$ satisfies in all the cases.
\item $\MINUS$ fails in four cases. 
For instance, if $\ev{P_1}{} = \Omega_1$ and $\ev{P_2}{} = \Omega_0$ then 
$\ev{((P_1 \UNION P_2) \MINUS P_2)}{} = \Omega_1 \cup \Omega_0$
whereas
$\ev{(P_1 \MINUS P_2)}{} = \Omega_1$.
\end{itemize}

\subsubsection{Axiom (k): $(P_1 \NEG (P_1 \AAND P3)) \equiv ((P_1 \NEG P_1) \UNION (P_1 \NEG P_3)$ }.
\begin{itemize}
\item $\DIFF$ fails in ten cases under bags semantics.
For instance, 
if $\ev{P_1}{} = \Omega_1$, $\ev{P_2}{} = \Omega_0$ and $\ev{P_3}{} = \Omega_2$ then 
$\ev{(P_1 \DIFF (P_1 \AAND P3))}{} = \Omega_2$
whereas
$\ev{((P_1 \DIFF P_1) \UNION (P_1 \DIFF P_3)}{} = (\Omega_2 \setminus \Omega_1) \cup \Omega_2$.

\item $\MINUS$ generates distinct solutions in twenty two cases, and fails under bag semantics in sixty five cases.
For instance, 
if $\ev{P_1}{} = \Omega_0$, $\ev{P_2}{} = \Omega_1$ and $\ev{P_3}{} = \Omega_1$ then 
$\ev{(P_1 \MINUS (P_1 \AAND P3))}{} = \Omega_0$
whereas\\
$\ev{((P_1 \MINUS P_1) \UNION (P_1 \MINUS P_3)}{} = \Omega_1$.

\end{itemize}

\subsubsection{Axiom (l): $(P_1 \NEG (P_1 \UNION P3)) \equiv ((P_1 \NEG P_1) \AAND (P_1 \NEG P_3)$ }.
\begin{itemize}
\item Under bag semantics, $\DIFF$ and $\MINUS$ fail in seven and forty eight cases respectively.
\item For instance, both operators fail when 
$\ev{P_1}{} = \Omega_0$, $\ev{P_2}{} = \Omega_0$ and $\ev{P_3}{} = \Omega_1$ 
such that  
$\ev{(P_1 \NEG (P_1 \UNION P3))}{} = \Omega_1$
and\\
$\ev{((P_1 \NEG P_1) \AAND (P_1 \NEG P_3)}{} = \Omega_1 \Join \Omega_1$.

\end{itemize}

\end{document}